\documentclass{llncs}
\usepackage{amsmath}
\usepackage{tikz}
\usepackage{enumerate}
\usetikzlibrary{arrows}
\usepackage[linesnumbered]{algorithm2e}
\tikzstyle{vertex}=[circle, draw, inner sep=0pt, fill=black, minimum size=5pt]
\newcommand{\vertex}{\node[vertex]}
\begin{document}

\title{Minimum Eccentricity Shortest Path Problem:\\ an Approximation Algorithm and\\ Relation with the k-Laminarity Problem
}

\author{\'Etienne Birmel\'e\inst{1} \and Fabien de Montgolfier\inst{2}
 \and L\'eo Planche\inst{1,2}}
\authorrunning{Birmel\'e, de Montgolfier, Planche} 

\institute{MAP5, UMR CNRS 8145, Univ. Sorbonne Paris Cit\'e\\
\email{etienne.birmele@parisdescartes.fr}
\and
IRIF, UMR CRNS 8243, Univ. Sorbonne Paris Cit\'e
  \email{fm@liafa.univ-paris-diderot.fr} \email{leo\_planche@liafa.univ-paris-diderot.fr}
  }
  
\maketitle           

\begin{abstract}

The Minimum Eccentricity Shortest Path (MESP) Problem consists in determining a shortest path (a path whose length is the distance between its extremities) of minimum eccentricity in a graph. It was introduced by Dragan and Leitert~\cite{Dragan2015} who described a linear-time algorithm which is an $8$-approximation of the problem. In this paper, we study deeper the double-BFS procedure used in that algorithm and extend it to obtain a linear-time $3$-approximation algorithm. We moreover study the link between the MESP problem and the notion of laminarity, introduced by V\"olkel \emph{et al} \cite{Volkel2016}, corresponding to its restriction to a diameter (\emph{i.e.} a shortest path of maximum length), and show tight bounds between MESP and laminarity parameters.

\end{abstract}

\section{Introduction}

For both graph classification purposes and applications, it is an important issue to determine to which extent a graph can be summarized by a path. Different path constructions and metrics to characterize how far the graph is from the constructed path can be used, for example path-decompositions and path-width \cite{Robertson83} or path-distance-decompositions and  path-distance-width \cite{Yamazaki1997}. Another approach, on which we focus in this article, is to characterize the graph by a spine defined by one of its paths. 

\bigskip

This problem was first studied in terms of domination, that is finding a path such that every vertex in the graph belongs to or has a neighbor in the path. Several graphs classes were defined in terms of dominating paths. \cite{Deogun1995} studies the graphs for which the dominating path is a diameter. \cite{Deogun2002} introduces dominating pairs, that is vertices such that every path linking them is dominating. Graphs such that short dominating paths are present in all induced subgraphs are characterized in \cite{Bacso2007}. Linear-time algorithms to find dominating paths or dominating vertex pairs were also developed for AT-free graphs \cite{Corneil95,Corneil99}.

\bigskip

Dominating paths  do not exist however in every graph and have no associated metric 
to measure the distance from the graph to the path. A natural extension of the notion of domination is the notion of $k$-coverage for a given integer $k$, defined by the fact that a path $k$-covers the graphs if every vertex is at distance at most $k$ from the path. The smallest $k$ such that a path $k$-covers the graph is then a metric as desired.  

\bigskip

In the present paper, we study the latter problem in which the covering path is required to be a shortest path between its end-vertices. It was introduced in \cite{Dragan2015} as the Minimum Eccentricity Shortest Path Problem, and shown to be linked to the minimum line distortion problem \cite{yan2007graph}.

\bigskip

The MESP problem is also closely related to the notion of $k$-laminar graphs introduced in \cite{Volkel2016}, in which the covering path is required to be a diameter. 

\bigskip

The MESP problem, as well as determining if a graph is $k$-laminar for a given $k$, are NP-hard \cite{Dragan2015,Volkel2016}. 
However, Dragan and Leitert \cite{Dragan2015} develop a 2-approximation algorithm for MESP of time complexity $\mathcal{O}(n^3)$, a 3-approximation algorithm in $\mathcal{O}(nm)$ and a linear 8-approximation. The latter is extremely simple as it consists in a double-BFS procedure. 

\subsection*{Roadmap}
In this paper, we introduce a different analysis of the double-BFS procedure and prove that it is in fact a 5-approximation algorithm, and that the bound is tight. We then develop the idea of this algorithm and reach a 3-approximation, which still runs in linear time. Finally, we establish bounds relating the MESP problem and the notion of laminarity.

\subsection*{Definitions and Notations}

Through this paper $G=(V,E)$ denotes a finite connected undirected graph.
A \emph{shortest path} between two vertices $u$ and $v$ is a path whose length is minimal among all $u,v$-paths.
This length (counting edges) is the \emph{distance} $d(u,v)$. 
Depending on the context, we consider a path either as a sequence, or as a set of vertices.
The \emph{distance} $d(v,S)$ between a vertex $v$ and a set $S$ is smallest distance between $v$ and a vertex from $S$.

The \emph{eccentricity} $ecc(S)$ of a set $S$ is the largest distance between $S$ and any vertex of $G$.

 The maximal eccentricity of any singleton $\{v\}$, or equivalently the largest distance between two vertices, denoted here $diam(G)$, is often called the diameter of the graph, but for clarity in this paper \emph{a diameter} is always a shortest path of maximum length, \emph{i.e.} a shortest path of length $diam(G)$, and not its length.

\pagebreak

\section{Double-BFS is a $5$-Approximation Algorithm}\label{section5k}

Let us define the problem we are interested in:

\begin{definition}[Minimum Eccentricity Shortest Path Problem (MESP)]
Given a graph $G$, find a shortest path $P$ such that, for every shortest path $Q$, $ecc(P)\leq ecc(Q)$.

\bigskip

$k(G)$ denotes the eccentricity of a MESP of $G$.
\end{definition}

\begin{theorem}[Dragan and Leitert \cite{Dragan2015}]
Computing $k(G)$ or finding a MESP are NP-complete problems.
\end{theorem}

It is therefore worth using polynomial-time approximation algorithms. We say that an algorithm is an $\alpha$-approximation of the MESP if every path output by this algorithm is a shortest path of eccentricity at most $\alpha k(G)$.

\bigskip

Double-BFS is a widely used tool for approximating $diam(G)$ \cite{draganBFS}. It simply consists in the following procedure:
\begin{enumerate}
\item Pick an arbitrary vertex $r$
\item Perform a BFS (Breadth-First Search) starting at $r$ and ending at $x$. $x$ is thus one of the furthest vertices from $r$.
\item Perform a BFS (Breadth-First Search) starting at $x$ and ending at $y$.
\end{enumerate}
The output of the algorithm is the path from $x$ to $y$, called a \emph{spread path}, while its extremities $(x,y)$ are called a \emph{spread pair}. 
A folklore result is that the distance between $x$ and $y$  2-approximates the diameter of $G$. As noted by Dragan and Leitert, Double-BFS may also be used for approximating MESP: they have shown in~\cite{Dragan2015} that any spread path is an 8-approximation of the MESP problem.

\bigskip

The first result of the present paper is that any spread path is in fact a $5$-approximation of the MESP problem and that the bound is tight. But before we prove this result (Theorem~\ref{lemma2}), let us give the key lemma used for proving our three theorems:

\begin{lemma}\label{lemma1}
Let G be a graph having a shortest path $v_0,v_1...v_t$  of eccentricity k.

Let P=$x_0,x_1,...x_s$ be a shortest path of G.
    
Let $i_{min}^P$ (resp. $i_{max}^P$) be the smallest (resp. largest) integer such that $v_{i_{min}^P}$ (resp. $v_{i_{max}^P}$) is at distance at most k of P.

\bigskip

For every integer $i$ such that $i_{min}^P \leq i\leq i_{max}^P$, $v_i$ is then at distance at most $2k$ from $P$.

\bigskip

Subsequently, every vertex $v$ of $G$ at distance at most $k$ from the subpath between $v_{i_{min}^P}$ and $v_{i_{max}^P}$ is at distance at most $3k$ of $P$.
\end{lemma}

One may think, at first glance, that this lemma looks similar to the following:
\begin{lemma}[from Dragan \emph{et al.} \cite{Dragan2015}]\label{lemmaDragan}
If $G$ has a shortest path of eccentricity at most $k$ from $s$ to $t$, then
every path $Q$ with $s$ in $Q$ and $d(s,t) \leq max_{v\in Q} d(s,v)$ has eccentricity at
most $3k$.
\end{lemma}

The difference lies in the fact that the $k$ in Lemma~\ref{lemmaDragan} is specific to the given couple of vertices $(s,t)$ while the $k$ in Lemma~\ref{lemma1} is global. 
On the other hand, Lemma~\ref{lemmaDragan} gives a bound on the eccentricity of a path with respect to the whole graph, while Lemma~\ref{lemma1} only guarantees an eccentricity for a defined subgraph.

\begin{proof}[of Lemma~\ref{lemma1}]
The second assertion of the lemma is straightforward given the first one. To prove the latter, we define, for all $l$ between $0$ and $s$, the subpath  $P_l=x_0,x_1...x_l$.

\bigskip

Let us show by induction on $l$ that for all $i$ between $i_{min}^{P_l}$ and $i_{max}^{P_l}$, $v_i$ is at distance at most $2k$ of $P_l$.\\

$\bullet$ $l=0$, $P_0=x_0$.

Using the triangle inequality:
\begin{equation} d(v_{i_{min}^{P_0}},v_{i_{max}^{P_0}})\leq d(v_{i_{min}^{P_0}},x_0)+d(x_0,v_{i_{max}^{P_0}}) \leq 2k \end{equation}

Hence, for all $i$ between $i_{min}^{P_0}$ and $i_{max}^{P_0}$, 
\begin{equation} d(v_{i_{min}^{P_0}},v_i)\leq k~or~d(v_{i_{max}^{P_0}},v_i)\leq k\end{equation}

The result is thus verified for $l=0$.\\ 

$\bullet$ Let $l$ in $(1...s)$ such that the property if verified for $l-1$.

For all $i$ between $i_{min}^{P_{l-1}}$ and $i_{max}^{P_{l-1}}$, $v_i$ is at distance at most $2k$ of $P_{l-1}$ by the induction hypothesis. Hence, $v_i$ is at distance at most $2k$ of $P_l$.

Moreover,
\begin{equation}d(v_{i_{max}^{P_{l-1}}},v_{i_{max}^{P_l}})\leq d(v_{i_{max}^{x_{l-1}}},v_{i_{max}^{x_l}})\end{equation}

and by the triangle inequality: 
\begin{equation} d(v_{i_{max}^{x_{l-1}}},v_{i_{max}^{x_l}}) \leq d(v_{i_{max}^{x_{l-1}}},x_{l-1})+d(x_{l-1},x_l)+d(x_l,v_{i_{max}^{x_l}}) \leq 2k+1\end{equation}

As the sub-path of P between $v_{i_{max}^{P_{l-1}}}$ and $v_{i_{max}^{P_l}}$ is a shortest path, it follows that for all $i$ between $i_{max}^{P_{l-1}}$ and $i_{max}^{P_l}$,
 \begin{equation} d(v_{i_{max}^{P_{l-1}}},v_i)\leq k~or~d(v_{i_{max}^{P_l}},v_i)\leq k,\end{equation} 
 
meaning that $v_i$ is at distance at most $2k$ of $P_{l-1}$ or of $x_l$.

\bigskip

A similar proof shows that for all $i$ between $i_{min}^{P_l}$ and $i_{min}^{P_{l-1}}$, $v_i$ is at distance at most $2k$ from $P_{l-1}$ or from $x_l$.
\bigskip

The property is verified by induction, and the lemma follows for $l=s$.
\end{proof} 

\begin{theorem}\label{lemma2}
	A double-BFS is a linear-time 5-approximation algorithm for the MESP problem.
\end{theorem}

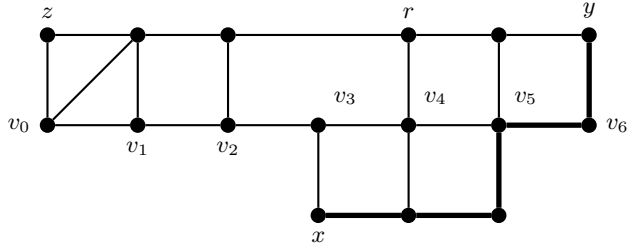
\begin{figure}

\centering
\begin{tikzpicture}[thick,scale=1.2,trans/.style={thick,<->,dashed}]

\vertex[label=left:$v_0$](0) at (0,0) {};
\vertex[label=below:$v_1$](1) at (1,0) {};
\vertex[label=below:$v_2$](2) at (2,0) {};
\vertex[](3) at (3,0) {};
\vertex[](4) at (4,0) {};
\vertex[](5) at (5,0) {};
\vertex[label=right:$v_6$](6) at (6,0) {};
\vertex[label=above:$r$](r) at (4,1) {};
\vertex[](10) at (1,1) {};
\vertex[](11) at (2,1) {};
\vertex[label=above:$z$](12) at (0,1) {};
\vertex[](13) at (4,-1) {};
\vertex[](14) at (5,1) {};
\vertex[](15) at (5,-1) {};
\vertex[label=below:$x$](x) at (3,-1) {};
\vertex[label=above:$y$](y) at (6,1) {};
\node (p3) at ( 3.3,0.3) {$v_3$};
\node (p4) at ( 4.3,0.3) {$v_4$};
\node (p4) at ( 5.3,0.3) {$v_5$};

\draw (0)--(1);
\draw (1)--(2);
\draw (2)--(3);
\draw (3)--(4);
\draw (4)--(5);
\draw[line width=2pt] (5)--(6);
\draw (2)--(11);
\draw (10)--(11);
\draw (11)--(r);
\draw (0)--(10);
\draw (1)--(10);
\draw (3)--(x);
\draw (r)--(4);
\draw[line width=2pt] (6)--(y);
\draw (12)--(0);
\draw (12)--(10);
\draw (13)--(4);
\draw[line width=2pt] (13)--(x);
\draw (14)--(r);
\draw (14)--(5);
\draw (14)--(y);
\draw[line width=2pt] (15)--(5);
\draw[line width=2pt] (15)--(13);

\end{tikzpicture}
\caption{The bound shown in Theorem~\ref{lemma2} is tight. Indeed the graph is such that $v_0,v_1...v_6$ is a shortest path of eccentricity 1. The vertex $z$ is at distance 5 from the shortest path (shown by thick edges) between $x$ and $y$ computed by double-BFS starting at $r$. \label{fig:tightness5k}}
\end{figure}

Before we prove it, notice that Figure~\ref{fig:tightness5k} shows that this bound is tight.

\begin{proof}
Let $k$ be $k(G)$,  $P = v_0,v_1...v_t$  be a MESP (its eccentricity is thus $k$), and $Q=x,...,y$ be the result of a double-BFS starting at some arbitrary vertex $r$, then reaching $x$, then reaching  $y$.
We shall prove that $Q$ is a $5k$-dominating path of $G$.

\bigskip

Let $i$ (resp. $j$) be such that $v_i$ (resp. $v_j$) is at distance at most $k$ of $r$ (resp. $x$). The following inequalities are verified:
\begin{equation}d(r,x)\geq d(r,v_t) \geq d(v_i,v_t)-d(r,v_i) \geq d(v_i,v_t)-k \end{equation}
\begin{equation}d(r,x) \leq d(r,v_i)+d(v_i,v_j)+d(v_j,x) \leq  d(v_i,v_j)+2k \end{equation} 

Combining those inequalities, 
\begin{equation} d(v_i,v_t)-3k \leq  d(v_i,v_j)\end{equation}

Similarly: 
\begin{equation}d(v_i,v_0)-3k \leq  d(v_i,v_j)\end{equation}

Therefore $v_j$ is at distance at most $3k$ of $v_0$ or $v_t$. Without loss of generality, assume that $v_j$ is at distance at most $3k$ of $v_0$.

\bigskip

Let $l$ be such that $v_l$ is at distance at most $k$ of $y$. We distinguish two cases:

\begin{enumerate}[(i)]
\item  \label{case1lemme2} $l\leq j$:

Then $y$ is at distance at most $5k$ of $x$. As $y$ is a vertex most distant from $x$, $x$ is a $5k$-dominating vertex of the graph. The lemma is then verified.

\item \label{case2lemme2} $l>j:$

Applying to $(x,y)$ the inequalities established at the beginning of the proof:
\begin{equation}d(v_j,v_t)-3k \leq  d(v_j,v_l)\end{equation}

As $l>j$, it follows that:
\begin{equation}d(v_l,v_t)\leq 3k\end{equation}

Figure~\ref{fig:5kconfig} shows the configuration of the graph in that case. 
The vertices at distance at most $k$ of a vertex $v_s$ such that $s\leq j$ (resp. $s\geq l$) are at distance at most $5k$ of $x$ (resp. $y$).

According to Lemma~\ref{lemma1}, every vertex $v$ of $G$ at distance at most $k$ of a vertex $v_s$ such that $s$ is between $j$ and $l$ is at distance at most $3k$ of any shortest path between $x$ and $y$. The lemma is thus verified.
\end{enumerate}

\end{proof}

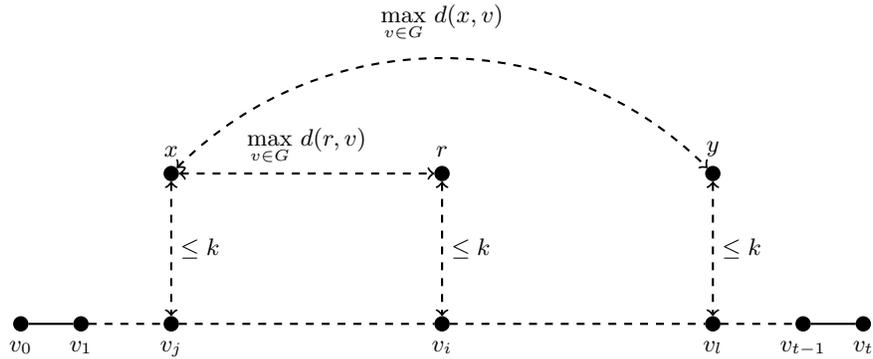
\begin{figure}

\centering
\begin{tikzpicture}[thick,scale=0.8,trans/.style={thick,<->,dashed}]

\vertex[label=below:$v_0$](0) at (0,0) {};
\vertex[label=below:$v_1$](1) at (1,0) {};
\vertex[label=below:$v_j$](j) at (2.5,0) {};
\vertex[label=below:$v_i$](i) at (7,0) {};
\vertex[label=below:$v_l$](l) at (11.5,0) {};
\vertex[label=below:$v_{t-1}$](n-1) at (13,0) {};
\vertex[label=below:$v_t$](n) at (14,0) {};
\vertex[label=above:$x$](x) at (2.5,2.5) {};
\vertex[label=above:$r$](r) at (7,2.5) {};
\vertex[label=above:$y$](y) at (11.5,2.5) {};

\draw (0)--(1);
\draw[dashed] (1)--(j);
\draw[dashed] (j)--(i);
\draw[dashed] (i)--(l);
\draw[dashed] (l)--(n-1);
\draw (n-1)--(n);
\draw[trans] (x)--(j) node[midway,right] {$\leq k$};
\draw[trans] (r)--(i) node[midway,right] {$\leq k$};
\draw[trans] (y)--(l) node[midway,right] {$\leq k$};
\draw[trans] (x)--(r) node[midway,above] {$\underset{v\in G}{\text{max}}$ $d(r,v)$};
\draw[trans] (x) to[bend left=45](y) node at (7,5) {$\underset{v\in G}{\text{max}}$ $d(x,v)$};
\end{tikzpicture}
\caption{Notations used in the proof of Theorem~\ref{lemma2} \label{fig:5kconfig}} 
\end{figure}

\section{A $3$-Approximation Algorithm}\label{section3k}

We show now that by using more BFS runs we may obtain a $3k$-approximation of MESP, still in linear time.\\ 

Let {\em bestPath} and {\em bestEcc} be global variables used as return values for the path and its eccentricity. {\em bestPath} stores a path and is uninitialized, and {\em bestEcc} is an integer initialized with $ |V(G)|$.

\begin{algorithm}[H]
	\KwData{$G$ graph, $x$,$y$ vertices of $G$, $step$ integer}
    Compute a shortest path $Q$ between $x$ and $y$\; \label{line1}
    Select a vertex $z$ of G most distant from $Q$\; \label{line2}
	\If{$d(Q,z)< bestEcc$}{
	$bestPath \gets Q$\;
	$bestEcc \gets d(Q,z)$\;
}
	\If{$step < 8$}{
	Algorithm3k($G$,$x$,$z$,$step+1$)\;
	Algorithm3k($G$,$y$,$z$,$step+1$)\;
	}
	\caption{Algorithm3k} 
\end{algorithm}

\begin{theorem}\label{key3}
	A $3$-approximation of the MESP Problem can be computed in linear time by considering a spread pair $(s,l)$ of $G$ and  running Algorithm3k($G$,$s$,$l$,$0$).
\end{theorem}

\begin{proof}[Correctness]
Let G be a graph admitting a shortest path P = $v_0,v_1...v_t$ of eccentricity k.

\medskip

Let $x$ and $y$ be any vertices of $G$, $Q_{x,y}$ a shortest path between $x$ and $y$. 
Define $i_{min}^{x,y}$ (resp. $i_{max}^{x,y}$) as the smallest (resp. largest) integer such that $v_{i_{min}^{x,y}}$ (resp. $v_{i_{max}^{x,y}}$) is at distance at most $k$ of $x$ or $y$. Then, by Lemma~\ref{lemma1}, 

\begin{equation} \label{Qdistance} \mbox{For all } j \mbox{ such that } i_{min}^{x,y} - k \leq j \leq i_{max}^{x,y} +k , \quad d(Q_{x,y},v_j) \leq 2k   \end{equation} 

Hence, if $i_{min}^{x,y} \leq k$ and  $i_{max}^{x,y} \geq t-k$, every vertex of $P$ is at distance at most $2k$ of $Q_{x,y}$ and, as $P$ is of eccentricity $k$, $Q_{x,y}$ is of eccentricity at most $3k$.

\medskip

Algorithm3k uses this implication to exhibit a pair $x,y$ such that $Q_{x,y}$ is of eccentricity at most $3k$. Indeed, in each recursive call, one of the following cases holds:

\begin{enumerate}
\item  the vertex $z$ selected at line $3$ is at distance at most $3k$ from $Q_{x,y}$. In that case, {\em bestPath} will be set to $Q_{x,y}$ unless it already contains a path of even better eccentricity. In any case, the result of the algorithm is a path of eccentricity at most $3k$.
\item the vertex $z$ is at a distance greater than $3k$ of $Q_{x,y}$.  Let $i_z$ be such that $v_{i_z}$ is at distance at most $k$ of $z$. Then, according to Equation~(\ref{Qdistance}),

\begin{equation}i_z \leq i_{min}^{x,y}-k~or~i_z \geq i_{max}^{x,y}+k\end{equation}

\begin{enumerate}
\item
Suppose that $i_z\geq i_{max}^{x,y}+k$. Then, in the case $d(v_{i_{min}^{x,y}},x)= k$, we get $i_{min}^{x,z}\leq i_{min}^{x,y}$ and $i_{max}^{x,z}\geq i_{max}^{x,y}+k$. And in the case  $d(v_{i_{min}^{x,y}},y) = k$ we get $i_{min}^{x,z}\leq i_{min}^{x,y} - k$ and $i_{max}^{x,z}\geq i_{max}^{x,y}$.

\item A similar reasoning can be applied if $i_z \leq i_{min}^{x,y}-k$, also yielding to $i_{min}^{x,z}\leq i_{min}^{x,y}$ and $i_{max}^{x,z}\geq i_{max}^{x,y}+k$ or $i_{min}^{x,z}\leq i_{min}^{x,y} - k$ and $i_{max}^{x,z}\geq i_{max}^{x,y}$.
\end{enumerate}
\end{enumerate}

Therefore, either the algorithm already found a path of eccentricity at most $3k$, or it makes one of its two new calls with a couple $(x',y')$ such that  the interval $[i_{min}^{x',y'},i_{max}^{x',y'}]$ contains $[i_{min}^{x,y},i_{max}^{x,y}]$ but has length increased by at least $k$.

\medskip

Consider now a spread pair $(s,l)$ for which Algorithm3k($G$,$s$,$l$,$0$) is run. It follows from case~(\ref{case1lemme2}) and  (\ref{case2lemme2}) of the proof of Theorem~\ref{lemma2} that
\begin{equation}i_{min}^{s,l}\leq 5k \mbox{ and } i_{max}^{s,l}\geq t-5k\end{equation}

At each of the recursive calls, if no path of eccentricity at most $3k$ has already been discovered, one of the new calls expands the interval $[i_{min}^{x,y},i_{max}^{x,y}]$ length by at least $k$, while containing the previous interval.
As the recursive calls are made until $step=8$, it follows that either a path of eccentricity $3k$ has been discovered, or one of the explored possibilities corresponds to eight extensions of size at least $k$ starting from $[i_{min}^{s,l},i_{max}^{s,l}]$.

In the latter case, Equation~(14) implies that the final couple of vertices $(x,y)$ fulfills 
 $i_{min}^{x,y} \leq k$ and $i_{max}^{x,y} \geq t-k$. 
Every vertex of $P$ is then of distance at most $2k$ of $Q_{x,y}$ and thus $Q_{x,y}$ is of eccentricity at most $3k$.

\end{proof}

\begin{proof}[Complexity]
The algorithm computes two BFS trees at line \ref{line1} and \ref{line2}, taking $\mathcal{O}(n+m)$ time. The rest of the operations is computed in constant time.

The recursivity width is 2 and, since it is first called with $step=0$, the recursivity length is 8. The algorithm is thus called 255 times. Therefore the total runtime of the algorithm is $\mathcal{O}(n+m)$.
\end{proof}

\begin{proof}[Tightness of the approximation]
Figure \ref{fig3} shows a graph for which the algorithm may produce a path of eccentricity $3k(G)$ (see caption).
\end{proof}

\begin{figure}
\centering
\begin{tikzpicture}[thick,scale=1.2,trans/.style={thick,<->,dashed}]

\vertex[label=left:$v_0$](0) at (0,0) {};
\vertex[label=below:$v_1$](1) at (1,0) {};
\vertex[label=below:$v_2$](2) at (2,0) {};
\vertex[label=above:$v_3$](3) at (3,0) {};
\vertex[label=below:$v_4$](4) at (4,0) {};
\vertex[label=below:$v_5$](5) at (5,0) {};
\vertex[label=below:$v_6$](6) at (6,0) {};
\vertex[label=above:$v_7$](7) at (0.66,1) {};
\vertex[label=above:$v_8$](8) at (1.66,1) {};
\vertex[label=above:$v_9$](9) at (3,1) {};
\vertex[label=above:$v_{10}$](10) at (4.33,1) {};
\vertex[label=above:$v_{11}$](11) at (5.33,1) {};
\vertex[label=below:$v_{12}$](12) at (3,-1) {};

\draw (0)--(1);
\draw (1)--(2);
\draw (2)--(3);
\draw (3)--(4);
\draw (4)--(5);
\draw (5)--(6);
\draw (0)--(7);
\draw (7)--(8);
\draw (8)--(9);
\draw (9)--(10);
\draw (10)--(11);
\draw (11)--(6);
\draw (3)--(12);
\draw (1)--(8);
\draw (2)--(9);
\draw (4)--(9);
\draw (5)--(10);

\end{tikzpicture}
\caption{Tightness of the bound shown in Theorem~\ref{key3}. The algorithm may indeed loop between the following couples of vertices : $(v_0,v_6),(v_0,v_{12}),(v_6,v_{12}),(v_0,v_{11}),(v_{11},v_{12}),(v_6,v_7),(v_7,v_{12}),(v_{11},v_7)$. Each time, it may choose a shortest path of eccentricity $3$ (passing through $v_8$ $v_9$ and $v_{10}$ whenever $v_{12}$ is not an endvertex of the path) while $v_0..v_3..v_6$ has eccentricity 1.\label{fig3}}
\end{figure}
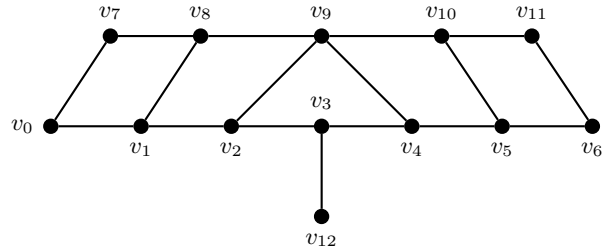

\section{Bounds between MESP and Laminarity}
In this section, we investigate the link between the MESP problem and the notion of laminarity introduced by V\"olkel \emph{et al.} in \cite{Volkel2016}. The study of the $k$-laminar graph class finds motivation both from a theoretical and practical point of view. On the theoretical side,  AT-free graphs form a well known  graph class introduced half a century ago by Lekkerkerker and Boland~\cite{LB62}, which contains many graph classes like co-comparability graphs. An AT-free graph admits a diameter all other vertices are adjacent with~\cite{corneil1997asteroidal}. It is then natural to extend this notion of dominating diameter. On the practical side, some large graphs constructed from reads similarity networks of genomic or metagenomic data appear to have a very long diameter and all vertices at short distance from it~\cite{Volkel2016}, and exhibiting the "best" diameter allows to better understand their structure.

\begin{definition}[laminarity]
A graph $G$ is 
\begin{itemize} 
\item $l$-laminar if $G$ has a diameter of eccentricity at most $l$.
\item $s$-strongly laminar if every diameter has eccentricity at most $s$.
\end{itemize}
$l(G)$ and $s(G)$ denote the minimal values of $l$ and $s$ such that $G$ is respectively $l$-laminar and $s$-strongly laminar. 
\end{definition}

A natural question about laminarity and MESP is to ask what link exists between them.

\begin{theorem}\label{propositionklaminarMESP}

For every graph $G$, 
\begin{align*}
k(G) & \leq l(G) \leq 4 k(G)-2 \\
k(G) & \leq s(G) \leq 4 k(G)
\end{align*}

Moreover, there exist three graph sequences  $(G_k)_{k\geq 1}$, $(H_k)_{k\geq 1}$ and $(J_k)_{k\geq 1}$  such that, for every $k$, 
\begin{itemize}
\item $k(G_k) = l(G_k) = s(G_k) = k$;
\item $k(H_k)=k$ and $l(H_k)=4k-2$;
\item $k(J_k)=k$ and $s(J_k)=4k$;
\end{itemize}

The bounds given by the inequalities are therefore tight.
\end{theorem}

\begin{proof}[$k(G) \leq l(G)$ and $k(G) \leq s(G)$]

Those inequalities are straightforward as every diameter is by definition a shortest path. The eccentricity of every diameter is therefore always greater than $k(G)$. 
\end{proof}

\begin{proof}[$s(G) \leq 4 k(G)$]\label{proofalpha4k}

Let $D=x_0,x_1,...x_s$ be a diameter of $G$ and $P=v_0,v_1...v_t$ a shortest path of eccentricity $k$. 
We shall show $ecc(D)\le 4k$. 
Let $z$ be any vertex of $G$. Since $ecc(P)=k$ there exists  a vertex $v_i$ of $P$ such that $d(z,v_i)\leq k$. Let us distinguish three cases:\\

$\bullet$ Case 1: there exists vertices $x_a$, $x_b$ of $D$ and $v_a$, $v_b$ of $P$ such that $a \leq i \leq b$ and
$d(v_a, x_a)\le k$ and $d(v_b, x_b)\le k$. Then by Lemma~\ref{lemma1}, $z$ is at distance at most $3k$ from any shortest path between $x_a$ and $x_b$, and thus at distance at most $3k$ of $D$.\\

$\bullet$ Case 2: there exists no vertex $v_a$ of $P$ with $a\le i$ and $d(v_a,D)\le k$

$\bullet$ Case 3: there exists no vertex $v_a$ of $P$ with $i\le a$ and $d(v_a,D)\le k$. 

\bigskip

Without loss of generality we focus on Case~2 (illustrated in Figure~\ref{fig:alpha4k}), which is symmetric with Case~3. Let $l$ (resp. $m$) be such that $v_l$ (resp. $v_m$) is at distance at most $k$ of $x_0$ (resp. $x_s$), assume $l\leq m$:
\begin{equation}
d(v_l,v_m) \geq d(x_0,x_s)-2k
\end{equation}

$D$ being a diameter,
\begin{equation}
d(x_0,x_s) \geq d(v_0,v_t) 
\end{equation}

By combining those inequalities,
\begin{equation}
d(v_l,v_m) \geq d(v_0,v_t) -2k
\end{equation}
\begin{equation}
d(v_l,v_m)\geq d(v_0,v_i)+d(v_i,v_l)+d(v_l,v_m)+d(v_m,v_t)-2k
\end{equation}
\begin{equation}\label{eq18}
2k \geq d(v_i,v_l)
\end{equation}

It follows that $z$ is at distance at most $4k$ of $x_0$.
\end{proof}

\begin{proof}[$l(G) \leq 4k(G)-2$]

Let $D=x_0,x_1,...x_s$ be a diameter of $G$ and $P=v_0,v_1...v_t$ a shortest path of eccentricity $k$.
We shall show that either $ecc(D) \leq 4k-2$ or $G$ contains a diameter $D'$ of eccentricity $3k$. If $P$ is a diameter we are done. Let us suppose from now it is of length at most $|D|-1$.

\bigskip

Let $z$ be any vertex of $G$ and $v_i$ a vertex of $P$ such that $d(z,v_i)\leq k$.
Let us distinguish the same three cases than in the proof that $s(G) \leq 4 k(G)$. The first case also leads to $d(z,D)\le 3k$. The second and third being symmetric, let us suppose there exists no vertex $v_j$ of $P$ at distance at most k of $D$ such that $j\leq i$. 

\bigskip

Let $v_l$ (resp. $v_m$) be a vertex of P at distance at most $k$ from $x_0$ (resp. $x_s$), clearly, 
\begin{equation}
d(v_l,v_m)\geq |D|-2k.
\end{equation}

Let us distinguish two subcases:\\

$\bullet$ Case 2.1: $d(v_l,v_m)>|D|-2k$, 
\begin{equation}
d(v_i,v_l)\leq d(v_0,v_t)-d(v_l,v_m) \leq (|D|-1) - (|D|-2k+1) \leq 2k-2
\end{equation}

It follows that $z$ is at distance at most $4k-2$ of $D$.

$\bullet$ Case 2.2: $d(v_l,v_m)=|D|-2k$

In this case, a path $D'=x_0,..v_l,v_{l+1},..v_m,..x_s$ is a diameter. Assuming $l\leq m$, Equation~\ref{eq18} in previous proof shows that:

\begin{equation}
d(v_i,v_l) \leq 2k
\end{equation}

and with a symmetrical reasoning,
\begin{equation}
d(v_m,v_t) \leq 2k
\end{equation}

It follows that any vertex $v$ of $G$ at distance at most $k$ of a vertex $v_a$ with $a\leq l$ (resp. $a \geq m$) is at distance at most $3k$ of $v_l$ (resp. $v_m$). Hence at distance at most $3k$ of $D'$.
$v_l,v_{l+1},..v_m$ being a subpath of $D'$, any vertex $v$ of $G$ at distance at most $k$ of a vertex $v_a$ with $a$ between $m$ and $t$ is at distance at most k of $D'$.
Finally, any vertex of $G$ is at distance at most $3k$ of $D'$.
\end{proof}

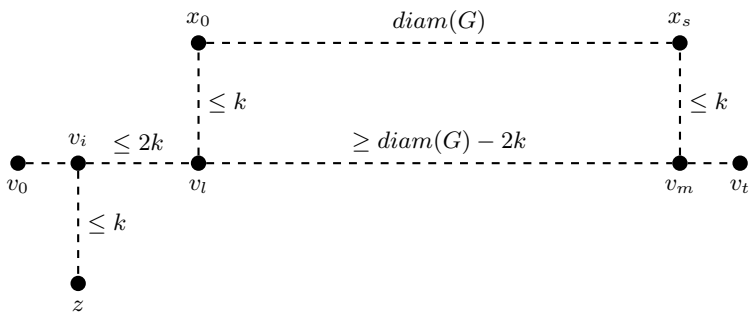
\begin{figure}

\centering
\begin{tikzpicture}[thick,scale=0.8,trans/.style={thick,<->,dashed}]

\vertex[label=below:$v_0$](0) at (0,0) {};
\vertex[label=above:$v_i$](i) at (1,0) {};
\vertex[label=below:$v_l$](l) at (3,0) {};
\vertex[label=below:$v_m$](m) at (11,0) {};
\vertex[label=below:$v_t$](t) at (12,0) {};
\vertex[label=below:$z$](z) at (1,-2) {};
\vertex[label=above:$x_0$](x0) at (3,2) {};
\vertex[label=above:$x_s$](xs) at (11,2) {};

\draw[dashed] (0)--(i);
\draw[dashed] (i)--(l) node[midway,above] {$\leq 2k$};
\draw[dashed] (l)--(m) node[midway,above] {$\geq diam(G)-2k$};
\draw[dashed] (m)--(t);
\draw[dashed] (i)--(z) node[midway,right] {$\leq k$};
\draw[dashed] (l)--(x0) node[midway,right] {$\leq k$};
\draw[dashed] (m)--(xs) node[midway,right] {$\leq k$};
\draw[dashed] (x0)--(xs) node[midway,above] {$diam(G)$};

\end{tikzpicture}
\caption{Notations used in Case 2 of the proof of Theorem~\ref{propositionklaminarMESP}\label{fig:alpha4k}} 
\end{figure}

\begin{proof}[Tightness of the bounds]

Consider the graph $G_k$  reduced to a path $P$ of length $4k$ to which a second path of length $k$ is attached in the middle. $P$ is then simultaneously the only diameter and the MESP, and it $k$-covers $G_k$ but doesn't $(k-1)$-cover it. Hence the inequalities $k(G) \leq l(G)$ and $k(G) \leq s(G)$ are tight.

\bigskip

 Figure~\ref{fig:tightness4} shows how to build the graph sequence $(J_k)_{k\ge 1}$ (only $J_1$ and $J_6$ are drawn). $J_k$ is a graph with a shortest path of eccentricity $k$ and a diameter of eccenticity $4k$. The inequality $s(G) \leq 4 k(G)$ is thus tight.

\bigskip

 Figure~\ref{fig:tightnessStrong} shows how to build the graph sequence $(H_k)_{k\ge 1}$ (only $H_1$, $H_2$ and and $H_6$ are drawn). $H_k$ is a graph with a shortest path of eccentricity $k$, while the unique diameter has eccenticity $4k-2$ ($H_1$ is a special case with two diameters). The inequality $l(G) \leq 4k(G)-2$ is therefore tight.
 \end{proof}

\begin{figure}[t!]
\centering
\begin{tikzpicture}[thick,scale=0.5,trans/.style={thick,<->,dashed}]

\vertex[label=below:$x_0$](1,0) at (1,0) {};
\vertex[](2,0) at (2,0) {};
\vertex[](3,0) at (3,0) {};
\vertex[](4,0) at (4,0) {};
\vertex[](5,0) at (5,0) {};
\vertex[](6,0) at (6,0) {};
\vertex[label=below:$v_0$](7,0) at (7,0) {};
\vertex[](8,0) at (8,0) {};
\vertex[](9,0) at (9,0) {};
\vertex[](10,0) at (10,0) {};
\vertex[](11,0) at (11,0) {};
\vertex[](12,0) at (12,0) {};
\vertex[](13,0) at (13,0) {};
\vertex[](14,0) at (14,0) {};
\vertex[](15,0) at (15,0) {};
\vertex[](16,0) at (16,0) {};
\vertex[](17,0) at (17,0) {};
\vertex[](18,0) at (18,0) {};
\vertex[label=right:$v_{2k}$](19,0) at (19,0) {};
\vertex[](2,1) at (2,1) {};
\vertex[](3,2) at (3,2) {};
\vertex[](4,3) at (4,3) {};
\vertex[](5,4) at (5,4) {};
\vertex[](6,5) at (6,5) {};
\vertex[](7,6) at (7,6) {};
\vertex[](8,6) at (8,6) {};
\vertex[](9,6) at (9,6) {};
\vertex[](10,6) at (10,6) {};
\vertex[](11,6) at (11,6) {};
\vertex[](12,6) at (12,6) {};
\vertex[](13,6) at (13,6) {};
\vertex[](14,6) at (14,6) {};
\vertex[](15,6) at (15,6) {};
\vertex[](16,6) at (16,6) {};
\vertex[](17,6) at (17,6) {};
\vertex[](18,6) at (18,6) {};
\vertex[](19,6) at (19,6) {};
\vertex[](20,6) at (20,6) {};
\vertex[](21,6) at (21,6) {};
\vertex[](22,6) at (22,6) {};
\vertex[](23,6) at (23,6) {};
\vertex[](24,6) at (24,6) {};
\vertex[label=below:$v_{4k}$](25,6) at (25,6) {};
\vertex[](14,7) at (14,7) {};
\vertex[](15,8) at (15,8) {};
\vertex[](16,9) at (16,9) {};
\vertex[](17,10) at (17,10) {};
\vertex[](18,11) at (18,11) {};
\vertex[](19,12) at (19,12) {};
\vertex[](20,12) at (20,12) {};
\vertex[](21,12) at (21,12) {};
\vertex[](22,12) at (22,12) {};
\vertex[](23,12) at (23,12) {};
\vertex[](24,12) at (24,12) {};
\vertex[label=right:$z$](25,12) at (25,12) {};
\vertex[](25,11) at (25,11) {};
\vertex[](25,10) at (25,10) {};
\vertex[](25,9) at (25,9) {};
\vertex[](25,8) at (25,8) {};
\vertex[](25,7) at (25,7) {};
\vertex[](25,6) at (25,6) {};
\vertex[](19,1) at (19,1) {};
\vertex[](19,5) at (19,5) {};
\vertex[](19,4) at (19,4) {};
\vertex[](19,3) at (19,3) {};
\vertex[](19,2) at (19,2) {};
\vertex[label=below:$x_{k+1}$](8,-1) at (8,-1) {};
\vertex[](9,-2) at (9,-2) {};
\vertex[](10,-3) at (10,-3) {};
\vertex[](11,-4) at (11,-4) {};
\vertex[](12,-5) at (12,-5) {};
\vertex[label=below:$x_{2k}$](13,-6) at (13,-6) {};
\vertex[](14,-6) at (14,-6) {};
\vertex[](15,-6) at (15,-6) {};
\vertex[](16,-6) at (16,-6) {};
\vertex[](17,-6) at (17,-6) {};
\vertex[](18,-6) at (18,-6) {};
\vertex[label=below:$x_{3k}$](19,-6) at (19,-6) {};
\vertex[](20,-6) at (20,-6) {};
\vertex[](21,-6) at (21,-6) {};
\vertex[](22,-6) at (22,-6) {};
\vertex[](23,-6) at (23,-6) {};
\vertex[](24,-6) at (24,-6) {};
\vertex[label=below:$x_{4k}$](25,-6) at (25,-6) {};

\node (p6) at ( 13.5,-0.5) {$v_{k}$};
\node (p6) at ( 19.7,5.6) {$v_{3k}$};
\node (p6) at ( 3.5,3.9) {\large{$k$}};
\node (p6) at ( 9.5,6.8) {\large{$k$}};
\node (p6) at ( 15.4,9.8) {\large{$k$}};
\node (p6) at ( 22,12.8) {\large{$k$}};

\foreach \x in {0,...,6}
    \foreach \y in {0,...,6} 
    {
     \vertex[](\x\y) at (\x+19,\y+6) {};
     \vertex[](\x\y) at (\x+7,\y) {};

}

\foreach \x in {0,...,6}
    \foreach \y in {0,...,\x} 
    {
     \vertex[](\x\y) at (\x+1,\y ) {};
     \vertex[](\x\y) at (\x+7, - \y) {};
     \vertex[](\x\y) at (\x+13, \y+6) {};

}

\foreach \x in {0,...,5}
    \foreach \y in {\x,...,5} 
    {
     \vertex[](\x\y) at (\x+19, - \y ) {};
     \vertex[](\x\y) at (\x+13, - \y  ) {};
}
\draw (1,0)--(7,6) ;
\draw (7,6)--(19,6);
\draw (13,0)--(19,-6);
\draw[very thick,green] (19,6)--(25,6);
\draw (13,6)--(19,12);
\draw (25,12)--(19,12);
\draw (25,12)--(25,6);
\draw (19,-6)--(19,6);
\draw[very thick,green] (19,0)--(19,6);
\draw[very thick,green] (7,0)--(19,0);
\draw[very thick,red](1,0)--(7,0);
\draw[very thick,red] (7,0)--(13,-6);
\draw (25,-6)--(19,0);
\draw[very thick,red] (25,-6)--(13,-6);
\draw (13,-6)--(13,6);
\draw (12,-5)--(12,6);
\draw (11,-4)--(11,6);
\draw (10,-3)--(10,6);
\draw (9,-2)--(9,6);
\draw (8,-1)--(8,6);
\draw (7,0)--(7,6);
\draw (6,0)--(6,5);
\draw (5,0)--(5,4);
\draw (4,0)--(4,3);
\draw (3,0)--(3,2);
\draw (2,0)--(2,1);
\draw (13,-1)--(18,-6);
\draw (13,-2)--(17,-6);
\draw (13,-3)--(16,-6);
\draw (13,-4)--(15,-6);
\draw (13,-5)--(14,-6);
\draw (19,-1)--(24,-6);
\draw (19,-2)--(23,-6);
\draw (19,-3)--(22,-6);
\draw (19,-4)--(21,-6);
\draw (19,-5)--(20,-6);
\draw (14,7)--(25,7);
\draw (15,8)--(25,8);
\draw (16,9)--(25,9);
\draw (17,10)--(25,10);
\draw (18,11)--(25,11);
\draw (14,7)--(14,6);
\draw (15,8)--(15,6);
\draw (16,9)--(16,6);
\draw (17,10)--(17,6);
\draw (18,11)--(18,6);
\draw (19,12)--(19,6);
\draw (20,12)--(20,6);
\draw (21,12)--(21,6);
\draw (22,12)--(22,6);
\draw (23,12)--(23,6);
\draw (24,12)--(24,6);
\draw[trans] (0.6,0.4)--(6.6,6.4);
\draw[trans] (7,6.4)--(13,6.4);
\draw[trans] (12.6,6.4)--(18.6,12.4);
\draw[trans] (19,12.4)--(25,12.4);

\vertex[](1,-5) at (1,-5) {};
\vertex[](2,-4) at (2,-4) {};
\vertex[](2,-5) at (2,-5) {};
\vertex[](3,-4) at (3,-4) {};
\vertex[](3,-5) at (3,-5) {};
\vertex[](3,-6) at (3,-6) {};
\vertex[](4,-3) at (4,-3) {};
\vertex[](4,-4) at (4,-4) {};
\vertex[](4,-5) at (4,-5) {};
\vertex[](4,-6) at (4,-6) {};
\vertex[](5,-6) at (5,-6) {};
\vertex[](5,-3) at (5,-3) {};
\vertex[](5,-4) at (5,-4) {};

\draw[green] (2,-5)--(4,-5)--(4,-4)--(5,-4);
\draw[red] (1,-5)--(2,-5)--(3,-6)--(5,-6);
\draw (1,-5)--(2,-4)--(3,-4)--(3,-6);
\draw (2,-4)--(2,-5);
\draw(4,-3)--(3,-4)--(4,-4)--(4,-3);
\draw(4,-6)--(4,-5)--(5,-6);
\draw(4,-3)--(5,-3)--(5,-4);
\draw(3,-5)--(4,-6);

\end{tikzpicture}
\caption{Proof that $s(G) \geq 4k(G)$. The red path $x_0,x_1,...x_{4k}$ is a diameter of length $4k$ and at distance $4k$ of $z$; while the green path $v_0,v_1,...v_{4k}$ is a shortest path (another diameter indeed) of eccentricity $k$. The large graph is $J_6$ (using the graph sequence $(J_k)_k$ from Theorem~\ref{propositionklaminarMESP}) and the small one on the bottom left is $J_1$. The other members of the sequence car easily be derived.  }\label{fig:tightness4}. 
\end{figure}
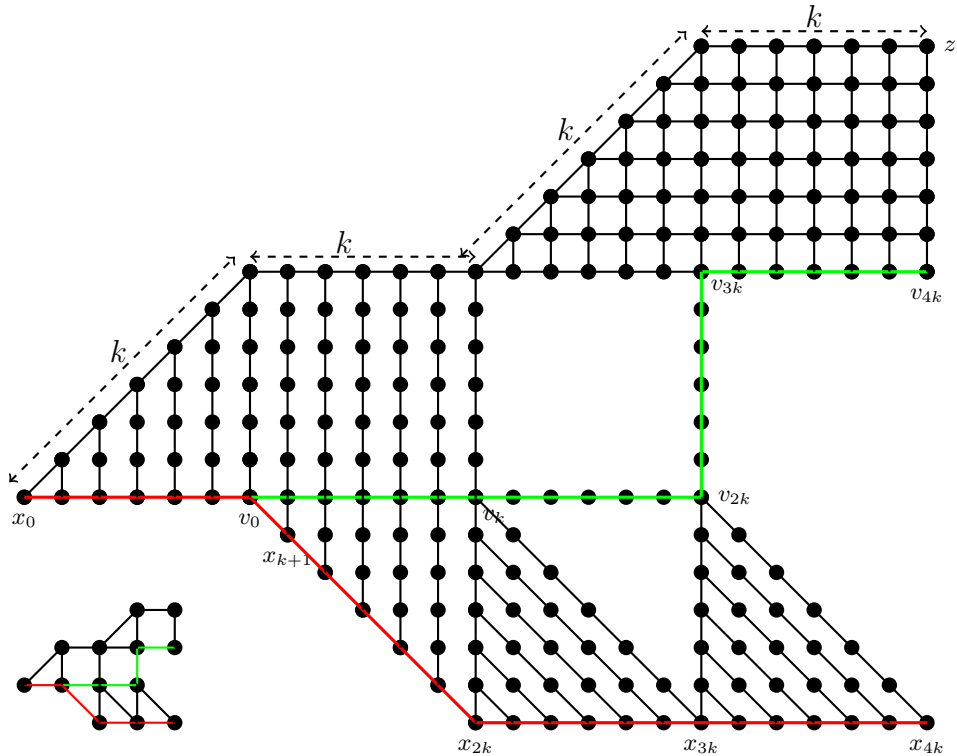

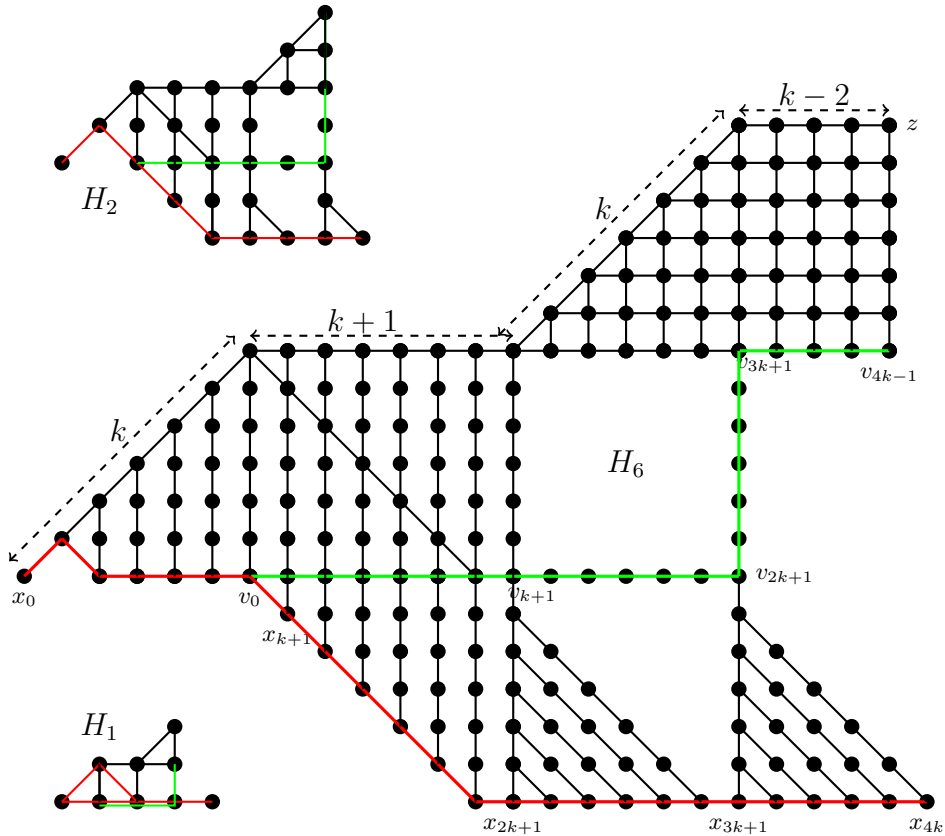
\begin{figure}[b!]
\centering
\begin{tikzpicture}[thick,scale=0.5,trans/.style={thick,<->,dashed}]

\vertex[label=below:$x_0$](0,0) at (0,0) {};
\vertex[](2,0) at (2,0) {};
\vertex[](3,0) at (3,0) {};
\vertex[](4,0) at (4,0) {};
\vertex[](5,0) at (5,0) {};
\vertex[label=below:$v_0$](6,0) at (6,0) {};
\vertex[](7,0) at (7,0) {};
\vertex[](8,0) at (8,0) {};
\vertex[](9,0) at (9,0) {};
\vertex[](10,0) at (10,0) {};
\vertex[](11,0) at (11,0) {};
\vertex[](12,0) at (12,0) {};
\vertex[](13,0) at (13,0) {};
\vertex[](14,0) at (14,0) {};
\vertex[](15,0) at (15,0) {};
\vertex[](16,0) at (16,0) {};
\vertex[](17,0) at (17,0) {};
\vertex[](18,0) at (18,0) {};
\vertex[label=right:$v_{2k+1}$](19,0) at (19,0) {};
\vertex[](1,1) at (1,1) {};
\vertex[](2,2) at (2,2) {};
\vertex[](3,3) at (3,3) {};
\vertex[](4,4) at (4,4) {};
\vertex[](5,5) at (5,5) {};
\vertex[](6,6) at (6,6) {};
\vertex[](7,6) at (7,6) {};
\vertex[](8,6) at (8,6) {};
\vertex[](9,6) at (9,6) {};
\vertex[](10,6) at (10,6) {};
\vertex[](11,6) at (11,6) {};
\vertex[](12,6) at (12,6) {};
\vertex[](13,6) at (13,6) {};
\vertex[](14,6) at (14,6) {};
\vertex[](15,6) at (15,6) {};
\vertex[](16,6) at (16,6) {};
\vertex[](17,6) at (17,6) {};
\vertex[](18,6) at (18,6) {};
\vertex[](19,6) at (19,6) {};
\vertex[](20,6) at (20,6) {};
\vertex[](21,6) at (21,6) {};
\vertex[](22,6) at (22,6) {};
\vertex[label=below:$v_{4k-1}$](23,6) at (23,6) {};
\vertex[](14,7) at (14,7) {};
\vertex[](15,8) at (15,8) {};
\vertex[](16,9) at (16,9) {};
\vertex[](17,10) at (17,10) {};
\vertex[](18,11) at (18,11) {};
\vertex[](19,12) at (19,12) {};
\vertex[](20,12) at (20,12) {};
\vertex[](21,12) at (21,12) {};
\vertex[](22,12) at (22,12) {};
\vertex[label=right:$z$](23,12) at (23,12) {};
\vertex[](23,11) at (23,11) {};
\vertex[](23,10) at (23,10) {};
\vertex[](23,9) at (23,9) {};
\vertex[](23,8) at (23,8) {};
\vertex[](23,7) at (23,7) {};
\vertex[](19,1) at (19,1) {};
\vertex[](19,5) at (19,5) {};
\vertex[](19,4) at (19,4) {};
\vertex[](19,3) at (19,3) {};
\vertex[](19,2) at (19,2) {};
\vertex[label=below:$x_{k+1}$](7,-1) at (7,-1) {};
\vertex[](8,-2) at (8,-2) {};
\vertex[](9,-3) at (9,-3) {};
\vertex[](10,-4) at (10,-4) {};
\vertex[](11,-5) at (11,-5) {};
\vertex[](12,-6) at (12,-6) {};
\vertex[label=below:$x_{2k+1}$](13,-6) at (13,-6) {};
\vertex[](14,-6) at (14,-6) {};
\vertex[](15,-6) at (15,-6) {};
\vertex[](16,-6) at (16,-6) {};
\vertex[](17,-6) at (17,-6) {};
\vertex[](18,-6) at (18,-6) {};
\vertex[label=below:$x_{3k+1}$](19,-6) at (19,-6) {};
\vertex[](20,-6) at (20,-6) {};
\vertex[](21,-6) at (21,-6) {};
\vertex[](22,-6) at (22,-6) {};
\vertex[](23,-6) at (23,-6) {};
\vertex[label=below:$x_{4k}$](24,-6) at (24,-6) {};
\node (p6) at ( 13.5,-0.5) {$v_{k+1}$};
\node (p6) at ( 19.7,5.6) {$v_{3k+1}$};
\node (p6) at ( 2.5,3.9) {\large{$k$}};
\node (p6) at ( 9,6.8) {\large{$k+1$}};
\node (p6) at ( 15.4,9.8) {\large{$k$}};
\node (p6) at ( 21,12.8) {\large{$k-2$}};

\foreach \x in {0,...,6}
    \foreach \y in {0,...,6} 
    {
     \vertex[](\x\y) at (\x+7,\y) {};

}

\foreach \x in {0,...,4}
    \foreach \y in {0,...,6} 
    {
     \vertex[](\x\y) at (\x+19,\y+6) {};
}

\foreach \x in {0,...,6}
    \foreach \y in {0,...,\x} 
    {
     \vertex[](\x\y) at (\x+7, - \y) {};
     \vertex[](\x\y) at (\x+13, \y+6) {};

}

\foreach \x in {0,...,5}
    { \vertex[](\x\x) at (\x+2,\x+1 ) {};
}

\foreach \x in {0,...,5}
    \foreach \y in {0,...,\x} 
    {
     \vertex[](\x\y) at (\x+2,\y ) {};
}

\foreach \x in {0,...,4}
    \foreach \y in {\x,...,4} 
    {
     \vertex[](\x\y) at (\x+19, -1 - \y ) {};
     \vertex[](\x\y) at (\x+13, -1 - \y  ) {};
}

\draw(6,6)--(12,0);
\draw (1,1)--(6,6) ;
\draw (6,6)--(19,6);
\draw[very thick,green] (19,6)--(23,6);
\draw (13,6)--(19,12);
\draw (23,12)--(19,12);
\draw (23,12)--(23,6);
\draw[very thick,green] (19,0)--(19,6);
\draw (19,-6)--(19,0);
\draw[very thick,green] (6,0)--(19,0);
\draw[very thick,red] (0,0)--(1,1);
\draw[very thick,red] (1,1)--(2,0);
\draw[very thick,red] (2,0)--(6,0);

\draw[very thick,red] (6,0)--(12,-6);
\draw[very thick,red] (24,-6)--(12,-6);
\draw (13,-6)--(13,6);
\draw (12,-6)--(12,6);
\draw (11,-5)--(11,6);
\draw (10,-4)--(10,6);
\draw (9,-3)--(9,6);
\draw (8,-2)--(8,6);
\draw (7,-1)--(7,6);
\draw (6,-0)--(6,6);
\draw (5,0)--(5,5);
\draw (4,0)--(4,4);
\draw (3,0)--(3,3);
\draw (2,0)--(2,2);
\draw (13,-1)--(18,-6);
\draw (13,-2)--(17,-6);
\draw (13,-3)--(16,-6);
\draw (13,-4)--(15,-6);
\draw (13,-5)--(14,-6);
\draw (19,-1)--(24,-6);
\draw (19,-2)--(23,-6);
\draw (19,-3)--(22,-6);
\draw (19,-4)--(21,-6);
\draw (19,-5)--(20,-6);
\draw (14,7)--(23,7);
\draw (15,8)--(23,8);
\draw (16,9)--(23,9);
\draw (17,10)--(23,10);
\draw (18,11)--(23,11);
\draw (14,7)--(14,6);
\draw (15,8)--(15,6);
\draw (16,9)--(16,6);
\draw (17,10)--(17,6);
\draw (18,11)--(18,6);
\draw (19,12)--(19,6);
\draw (20,12)--(20,6);
\draw (21,12)--(21,6);
\draw (22,12)--(22,6);
\draw[trans] (-0.4,0.4)--(5.6,6.4);
\draw[trans] (6,6.4)--(13,6.4);
\draw[trans] (12.6,6.4)--(18.6,12.4);
\draw[trans] (19,12.4)--(23,12.4);
\node (lab6) at ( 16,3) {\large{$H_6$}};
\node (lab2) at ( 2,-4) {\large{$H_1$}};
\node (lab1) at ( 2,10) {\large{$H_2$}};

\vertex[](1,-6) at (1,-6) {};
\vertex[](2,-6) at (2,-6) {};
\vertex[](3,-6) at (3,-6) {};
\vertex[](4,-6) at (4,-6) {};
\vertex[](5,-6) at (5,-6) {};
\vertex[](2,-5) at (2,-5) {};
\vertex[](3,-5) at (3,-5) {};
\vertex[](4,-5) at (4,-5) {};
\vertex[](4,-4) at (4,-4) {};

\draw (2,-6)--(2,-5)--(4,-5);
\draw[red] (1,-6)--(5,-6);
\draw[red] (1,-6)--(2,-5)--(3,-6);
\draw (4,-5)--(4,-4)--(3,-5);
\draw(2,-6)--(2,-5)  (3,-6)--(3,-5);
\draw[green](2,-6.1)--(4,-6.1)--(4,-5);

\vertex[](1,11) at (1,11) {};
\vertex[](2,12) at (2,12) {};
\vertex[](3,13) at (3,13) {};
\vertex[](3,12) at (3,12) {};
\vertex[](3,11) at (3,11) {};
\vertex[](4,10) at (4,10) {};
\vertex[](4,11) at (4,11) {};
\vertex[](4,12) at (4,12) {};
\vertex[](4,13) at (4,13) {};
\vertex[](5,9) at (5,9) {};
\vertex[](5,10) at (5,10) {};
\vertex[](5,11) at (5,11) {};
\vertex[](5,12) at (5,12) {};
\vertex[](5,13) at (5,13) {};
\vertex[](6,9) at (6,9) {};
\vertex[](6,10) at (6,10) {};
\vertex[](6,11) at (6,11) {};
\vertex[](6,12) at (6,12) {};
\vertex[](6,13) at (6,13) {};
\vertex[](7,9) at (7,9) {};
\vertex[](7,11) at (7,11) {};
\vertex[](7,14) at (7,14) {};
\vertex[](7,13) at (7,13) {};
\vertex[](8,9) at (8,9) {};
\vertex[](8,10) at (8,10) {};
\vertex[](8,11) at (8,11) {};
\vertex[](8,12) at (8,12) {};
\vertex[](8,13) at (8,13) {};
\vertex[](8,14) at (8,14) {};
\vertex[](8,15) at (8,15) {};
\vertex[](9,9) at (9,9) {};
\draw[red] (1,11)--(2,12)--(5,9)--(9,9);
\draw[green](3,11)--(8,11)--(8,15);
\draw(2,12)--(3,13)--(5,11)--(5,9);
\draw(3,11)--(3,13)--(6,13)--(8,15)--(8,13)--(7,13)--(7,14);
\draw (4,10)--(4,13) (5,13)--(5,9);
\draw (7,13)--(6,13)--(6,9);
\draw (6,10)--(7,9) (8,10)--(9,9) (7,14)--(8,14);
\draw (8,11)--(8,9);

\end{tikzpicture}
\caption{Proof that $l(G) \geq 4k(G)-2$. It is a graph sequence $(H_k)_k$, using the notation from Theorem~\ref{propositionklaminarMESP}. For $k\ge 2$, the red path $x_0,x_1,...x_{4k}$ is the unique diameter. Its length is $4k$ and it is at distance $4k-2$ of $z$. The green path $v_0,v_1,...v_{4k-1}$ is a shortest path of length $4k-1$ and of eccentricity $k$. Graphs $H_2$ and $H_6$ are drawn but all graphs $H_k$, $k\ge 2$ can be derived from the pattern of $H_6$. The small graph on the bottom left is the special case $H_1$ who do not follow this pattern. It admits exactly two diameters, both of eccentricity 2 (red), and a shortest path of eccentricity 1 (green). }\label{fig:tightnessStrong}.
\end{figure}
\section{Conclusion}

We have investigated the Minimum Eccentricity Shortest Path problem for general graphs and proposed a linear time algorithm computing a $3$-approximation. The algorithm is a 2-recursive function with constant recursivity depth, launching two BFSs each time, thus taking linear time. Additionally, we've established some tight bounds linking the MESP parameter $k(G)$ and the k-laminarity parameters $s(G)$ and $l(G)$.

\bigskip
On improving the current approximation algorithms, the following remark should be noted. Our algorithm is confined in finding a good pair of vertices in the graph, and the shortest path between them is then picked arbitrarily. By doing so, we are unlikely to get a better result than a $3$-approximation. Indeed as shown by \cite{Dragan2015} there exist graphs for which the MESP solution is a path of eccentricity $k$ between two vertices $s$ and $t$ such that some other shortest paths between $s$ and $t$ have an eccentricity of exactly $3k$. 

\bigskip
About laminarity parameters, computing $l(G)$ is NP-complete, while computing  $s(G)$ can be done in $O(	n^2m\log n)$ time~\cite{Volkel2016}. It may be interesting to design an approximation algorithm, \emph{i.e} producing a diameter of eccentricity at most $\alpha s(G)$ or $\beta l(G)$. Linear-time algorithms like BFS cannot be used however, since we do not know how to compute $diam(G)$  faster than a  matrix product, and even surlinear approximation are studied~\cite{aingworthDiam}. Different techniques than the ones used here must therefore be employed.

\bibliographystyle{splncs03}
\bibliography{references}

\end{document}